\title{Private Polynomial Computation from\\Lagrange Encoding}
\documentclass[10pt]{IEEEtran}

\usepackage[numbers,sort]{natbib}
\usepackage{bm}

\usepackage{amssymb}

\usepackage{amsmath}

\usepackage{amsthm}

\usepackage{mathtools}

\newtheorem{theorem}{Theorem}

\newtheorem{lemma}[theorem]{Lemma}

\newtheorem{example}[theorem]{Example}
\newtheorem{proposition}[theorem]{Proposition}

\usepackage{xcolor}
\usepackage{hyperref}
\hypersetup{
	colorlinks,
	linkcolor={blue!100!black},
	citecolor={blue!100!black},
	urlcolor={blue!80!black}
}

\usepackage{algorithm}
\usepackage{algpseudocode}

\newcommand{\Fq}{\mathbb{F}_{q}}

\newcommand{\bF}{\mathbb{F}}

\newcommand{\cA}{\mathcal{A}}
\newcommand{\cB}{\mathcal{B}}
\newcommand{\cC}{\mathcal{C}}
\newcommand{\cD}{\mathcal{D}}
\newcommand{\cE}{\mathcal{E}}

\newcommand{\cI}{\mathcal{I}}

\newcommand{\cP}{\mathcal{P}}

\newcommand{\cR}{\mathcal{R}}
\newcommand{\cS}{\mathcal{S}}
\newcommand{\cT}{\mathcal{T}}

\newcommand{\bolda}{\textbf{a}}
\newcommand{\boldb}{\textbf{b}}
\newcommand{\boldc}{\textbf{c}}
\newcommand{\boldd}{\textbf{d}}

\newcommand{\bolds}{\textbf{s}}
\newcommand{\boldt}{\textbf{t}}

\newcommand{\boldv}{\textbf{v}}

\newcommand{\boldx}{\textbf{x}}
\newcommand{\boldy}{\textbf{y}}

\newcommand{\boldA}{\textbf{A}}
\newcommand{\boldB}{\textbf{B}}
\newcommand{\boldC}{\textbf{C}}

\newcommand{\boldG}{\textbf{G}}

\newcommand{\boldX}{\textbf{X}}
\newcommand{\boldY}{\textbf{Y}}

\DeclareMathOperator{\spn}{span}
\DeclareMathOperator{\lcm}{lcm}

\newcommand{\boldalpha}{\boldsymbol{\alpha}}
\newcommand{\boldbeta}{\boldsymbol{\beta}}

\newcommand{\boldepsilon}{\boldsymbol{\epsilon}}

\newcommand{\boldrho}{\boldsymbol{\rho}}

\newcommand{\boldpsi}{\boldsymbol{\psi}}

\DeclareMathOperator{\image}{Im}

\DeclareSymbolFont{bbold}{U}{bbold}{m}{n}
\DeclareSymbolFontAlphabet{\mathbbold}{bbold}
\newcommand{\1}{\mathbbold{1}}

\author{\textbf{Netanel Raviv}$^\star$~and~\textbf{David A.~Karpuk}$^\dagger$\\
	\IEEEauthorblockA{\normalsize {$^\star$Department of Electrical Engineering, California Institute of Technology, Pasadena, CA 91125, USA.\\
			$^\dagger$Departamento de Matem\'{a}ticas, Universidad de los Andes, Bogot\'{a}, Colombia.}
		\thanks{Parts of this work were presented at the International Symposium on Information Theory (ISIT), Vail, CO, USA, 2018.}}}

\begin{document}
\maketitle
\begin{abstract}
Private computation is a generalization of private information retrieval, in which a user is able to compute a function on a distributed dataset without revealing the identity of that function to the servers. In this paper it is shown that Lagrange encoding, a powerful technique for encoding Reed-Solomon codes, enables private computation in many cases of interest. In particular, we present a scheme that enables private computation of polynomials of any degree on Lagrange encoded data, while being robust to Byzantine and straggling servers, and to servers colluding to attempt to deduce the identities of the functions to be evaluated. Moreover, incorporating ideas from the well-known Shamir secret sharing scheme allows the data itself to be concealed from the servers as well. Our results extend private computation to high degree polynomials and to data-privacy, and reveal a tight connection between private computation and coded computation.
\end{abstract}
\thispagestyle{empty}
\section{Introduction}\label{section:intro}
\subsection{Private Information Retrieval}
Private Information Retrieval (PIR) refers to the process of downloading a file from a database, without revealing to the database which file is being downloaded.  PIR was originally introduced in the seminal work \cite{chor} from the perspective of Computer Science, where the goal is to construct PIR scheme with minimal communication cost.  Much work has been done on PIR from this point of view, both under information-theoretic constraints~\cite{Yekhanin2} and computational ones~\cite{COMPsurvey}, of which we also briefly mention \cite{Gasarch,Yekhanin}.

Recently, due to rising interest in large scale distributed storage systems and emerging quantum attacks on computationally secure protocols, PIR has experienced a flurry of research from the Information Theory community.  In the Information-Theoretic formulation of the PIR problem, the primary performance metric of a PIR scheme is the download rate, which refers to the amount of downloaded bits per one information bit. The capacity of PIR, given a specific system setup, is the maximum possible download rate. Commonly, one assumes that files of length $K$ are distributed over $N$ different servers using an $[N,K]$ \textit{Maximum Distance Separable} (MDS) storage code, and every $T$ out of these $N$ servers can collude to try to deduce the identity of the file being downloaded.  The capacity of several variations on the PIR problem is now known, for example for $K = T = 1$ by \cite{SunJafar}, for $K = 1$ and $T\geq 1$ by \cite{SunJafarColluding}, and for $K\geq 1$ and $T = 1$ by \cite{BanawanUlukusCoded}.  Additional interesting and practically relevant variants on PIR appear when one accounts for unreliability of the servers, whether they be unresponsive, malicious, or curious as to the contents of the data they are storing.  Variants along these lines have been studied and capacity expressions have been derived in \cite{SunJafarColluding,BanawanUlukusByzantine,CrossSubspace}.

Many open questions about PIR remain, for example the PIR capacity is unknown when $K,T>1$, that is, when a non-trivial storage code is used and we have non-trivial server collusion; less is known when $K,T>1$ and we try to account for unreliable servers.  Given the difficulty of proving optimality in such generality, one is often content with achieving the more modest goal of constructing PIR schemes with large download rate and ignoring questions of capacity.  Algebraic constructions of PIR schemes for coded data with colluding servers appeared in \cite{RazanSalim,StarProduct}, which were then generalized in \cite{Karpuk2} to account for unresponsive and Byzantine servers.  Additional work along these lines has been done in \cite{transitive,NorwegiansNonMDS} for non-MDS storage codes.  These schemes generally fall under the umbrella of ``one-shot schemes'', a term recently christened in \cite{RafaelSalim}, where it was also shown how to lift such schemes to obtain schemes which achieve capacity in certain cases.  Reed-Solomon and Reed-Muller codes have been especially useful in such constructions, because of their compatibility with the star product between two linear codes (see Section~\ref{section:starproduct}).

\subsection{Private Computation}
Private Computation (PC) is a generalization of PIR wherein the user wants to not just privately download a file from the database, but privately compute an arbitrary function of it.  In PC, privacy refers to hiding the identity of the function to be computed.  One recovers the PIR problem by specifying to functions given by coordinate projections. The principal performance metric for PC protocols is the \textit{PC-rate} (or simply, rate), which is the ratio between the number of desired function evaluations and the total number of function evaluations downloaded.

For functions which are linear combinations of the files, PC was studied in \cite{PFR,SunJafarPC} for uncoded databases, and in \cite{ObeadKliewer,ObeadLinRosnesKliewer} for coded databases.  The case of non-linear functions, and especially polynomial functions of degree larger than one, was studied by the second author of the current work in \cite{Karpuk}, where a PC scheme was constructed for polynomial functions on systematically coded databases.

\subsection{Coded Computation}
The term \textit{Coded Computation} (CC) broadly refers to a family of techniques in which redundancy is added to datasets, in order to alleviate various issues that arise in distributed computations. Motivated by recent applications in machine learning, studies in Coded Computation have focused, e.g., on accelerating distributed tasks, combating malicious interventions, providing various forms of privacy, alleviating the communication load in iterative algorithms, and more. 

The study of this topic has been highly prolific in recent years, and works on the topic have mostly been task-specific. Typical tasks of interest include matrix multiplication~\cite{PolynomialCodes,ShortDot}, accelerating gradient descent algorithms~\cite{CycMDSandExp,GradientCoding1, GradientCodingSalman}, communication reduction~\cite{TradeoffComp}, and data shuffling~\cite{Shuffling1,Shuffling2}. Yet, coded computing of general polynomials has been addressed only recently in~\cite{Lagrange}, which is tightly connected to our results. In~\cite{Lagrange} it was shown that coding the data by using the well-known Lagrange polynomials can amend issues of resiliency, security, and privacy in many tasks of interest. 

\subsection{Current Contributions}
We present a Private Computation scheme for the evaluation of degree $G$ polynomials on $K$ data vectors $\boldx_k\in \Fq^M$, encoded using an $[N,K+E]$ Reed-Solomon code, which hides the identity of the polynomials to be computed from any~$T$ colluding servers, and hides the contents of the data vectors $\boldx_k$ from any $E$ colluding servers.  The scheme is robust against any $P$ stragglers and any $A$ adversaries, and has PC rate
\begin{alignat}{1}\label{mainrate}
R = \;&\frac{N-(G(K+E-1)+T+P+2A)}{N}\cdot\nonumber\\
&\frac{K}{G(K+E-1)+1},
\end{alignat}
which clearly requires that $N>(G(K+E-1)+T+P+2A)$. We provide a non-trivial example which illustrates the scheme construction.  Lastly, we propose an alternative scheme in the case that $P = A = E = 0$, which employs systematic data encoding and achieves a rate of
\[
R = \frac{\min\{N-(G(K-1)+T),K\}}{N}.
\]
This second scheme construction first appeared in \cite{Karpuk}. We provide some basic analysis as to the relative performance of the two schemes.

The scheme construction borrows ideas from Coded Computation, especially those of \cite{Lagrange}.  Data storage is realized by evaluating interpolating polynomials, which allows the evaluation of polynomials on encoded data to be viewed as the evaluation of a single variable polynomial.  This single-variable polynomial is the composition of two polynomial functions of known degree, and its degree can therefore be calculated explicitly and serves as crucial knowledge for the scheme construction and rate calculation.

Our techniques are applicable in many real-world scenarios. Beyond direct applications in distributed computations over finite fields, such as large scale matrix multiplication, one can also apply our techniques in real-number scenarios with minor additional effort, by quantization and embedding of real values into a large enough finite field. Such quantization techniques were employed in machine learning scenarios, e.g., in distributed linear regression in~\cite{Lagrange}, and in distributed logistic regression in~\cite{Jinhyun}.

\subsection{Comparison with Previous Work}
The scheme construction presented here also generalizes some ideas from \cite{Karpuk2} to deal with non-linear functions and data privacy.  In the case of $E = 0$ and $G = 1$ the current scheme achieves a rate of
\begin{equation}\label{mainPIRrate}
R = \frac{N-(K+T+P+2A-1)}{N}.
\end{equation}
If the functions to be computed are all distinct coordinate projections, then we reduce to the case of the PIR problem, and the rate of (\ref{mainPIRrate}) matches that of \cite{Karpuk2}.  If one further assumes that $P = A = 0$, one achieves a rate of
\begin{equation}\label{PLCrate}
R = \frac{N-(K+T-1)}{N}
\end{equation}
which is that of \cite{StarProduct}.  The rate (\ref{PLCrate}) is also the asymptotic rate (as the number of files $M\rightarrow\infty$) of \cite{ObeadLinRosnesKliewer} when $T = 1$, which studies linear computations of coded databases.  In all cases where $E = 0$ and the PIR or PC capacity $\mathsf{C}$ is known, the current scheme achieves the asymptotic capacity as the number of files grows, that is, $R = \lim_{M\rightarrow\infty}\mathsf{C}$.  

The capacity for the case $E > 0$ is subtler.  The case of $K = 1$, $G = 1$, $P  = A = 0$ is that of the recent work \cite{CrossSubspace}, which studies the PIR problem under the constraints of $T$-function privacy and $E$-data privacy, for storage systems in which every server stores an amount of data which is comparable to the entire dataset.  In \cite{CrossSubspace}, the authors show that
\[
\lim_{M\rightarrow\infty}\mathsf{C} = \left\{\begin{array}{cc}
\frac{N-(E+T)}{N} & \text{if } N > E+T \\
0 & \text{if } N \leq E + T
\end{array}\right.
\]
where $\mathsf{C}$ is the PIR capacity of this setting, and use a technique they deem \emph{Cross Subspace Alignment} to construct an explicit scheme which achieves the above rate.  The rate $R$ we achieve in this scenario is
\[
R = \frac{N - (E+T)}{N}\cdot \frac{1}{E+1}
\]
whenever $N > E+T$, which is strictly worse than that of \cite{CrossSubspace} when $E>0$.  This essentially stems from the fact that it is not clear how to align the noise terms arising from the function-privacy randomness and data-privacy randomness as in \cite{CrossSubspace} for non-linear functions.  Thus our PC scheme leaves room for improvement in the case of non-trivial data privacy.   

\section{Preliminaries}\label{section:preliminaries}
\subsection{Notation}
Throughout the current paper, we adopt the following notation.  For a given integer $N$, define $[N]$ to be the set $\{1,2,\ldots,N\}$ of the first $N$ natural numbers.  We use lowercase letters $a,b,c,\ldots$ for scalars or indexing variables, and uppercase letters $A,B,C,\ldots$ for system parameters, such as the number of servers, dimension of the storage code, etc.  Boldface lowercase letters $\bolda,\boldb,\boldc,\ldots$ are reserved for vectors, and boldface uppercase letters $\boldA,\boldB,\boldC,\ldots$ for matrices.  Calligraphic letters $\cA,\cB,\cC,\ldots$ are used to denote sets, including linear codes.  The symbol $\Fq$ is used to denote the finite field of cardinality $q$. 
Lowercase Greek letters such as $\phi,\psi,\rho,\ldots$ will generally be used to denote functions, and $\phi\circ\psi$ denotes the composition of $\phi$ and $\psi$, that is, $(\phi\circ\psi)(x) = \phi(\psi(x))$.

An $[N,K,D]_q$ code $\cC$ is a linear code over $\Fq$ with length $N$, dimension $K$, and minimum distance $D$.  We will often omit the $D$ or $q$ if they are unimportant or clear from context.  Recall that $\cC$ is MDS (Maximum Distance Separable) if $D = N-K+1$.

We let $\cP_{M,G}$ denote the set of all polynomials with coefficients in $\Fq$ in $M$ variables with total degree at most~$G$, that is,
\[
\cP_{M,G} = \{\phi\in \Fq[X_1,\ldots,X_M] : \deg(\phi)\leq G\}.
\]
Note that $\cP_{M,G}$ is a finite-dimensional vector space over $\Fq$ and therefore supports a uniform distribution.  The parameter $M$ will be of less interest to us in general and hence we will often write $\cP_G$ for $\cP_{M,G}$.

\subsection{Reed-Solomon Codes}
Given parameters $N$, $K$, and $q\geq N$, we construct the $[N,K]$ \emph{Reed-Solomon} code over $\Fq$ as follows.  Let $\boldalpha = (\alpha_1,\ldots,\alpha_N)$ be $N$ distinct elements of $\Fq$, and let
\[
\Fq[z]^{<K} = \{\phi\in \Fq[z] : \deg(\phi) \leq K-1\}
\]
We define the Reed-Solomon code $\cR\cS_K(\boldalpha)$ associated with this data to be the image of the evaluation map
\[
ev: \Fq[z]^{<K} \rightarrow \Fq^N,\quad ev(\phi) = (\phi(\alpha_1),\ldots,\phi(\alpha_N))
\]
It is a classical result that $\cR\cS_K(\boldalpha)$ is an MDS code.

We will be particularly interested in the encoding of Reed-Solomon codes.  If one chooses a basis $\phi_1,\ldots,\phi_K$ of the vector space $\Fq[z]^{<K}$, then a generator matrix (and therefore an encoding map) is given by
\[
\boldG = \begin{pmatrix}
\phi_1(\alpha_1) & \phi_1(\alpha_2) & \cdots & \phi_1(\alpha_N) \\
\phi_2(\alpha_1) & \phi_2(\alpha_2) & \cdots & \phi_2(\alpha_N) \\
\vdots & \vdots & \ddots & \vdots \\
\phi_K(\alpha_1) & \phi_K(\alpha_2) & \cdots & \phi_K(\alpha_N)
\end{pmatrix}.
\]
The following two choices of generator matrices will be the most important for our purposes.

\subsubsection{Canonical Encoding} The basis $1,z,\ldots,z^{K-1}$ of $\Fq[X]^{<K}$ gives rise to the Vandermonde generator matrix
\[
\boldG = \begin{pmatrix}
1 & 1 & \cdots & 1 \\
\alpha_1 & \alpha_2 & \cdots & \alpha_N \\
\vdots & \vdots & \ddots & \vdots \\
\alpha_1^{K-1} & \alpha_2^{K-1} & \cdots & \alpha_N^{K-1}
\end{pmatrix}
\]
Another way to view this encoding map is to map a message $\bolda = (a_1,\ldots,a_K)\in \Fq^{K-1}$ to the polynomial $\phi_{\bolda} = \sum_{i = 0}^{K-1}a_{i+1}z^i$.  The codeword is then $(\phi_{\bolda}(\alpha_1),\ldots,\phi_{\bolda}(\alpha_N)) = \bolda  \boldG$.

\subsubsection{Lagrange Encoding}\label{sec:lagrange_encode}\cite{Lagrange} In this case we choose a basis of $\Fq[z]^{<K}$ to consist of interpolating polynomials.  More specifically, choose some $\boldbeta = (\beta_1,\ldots,\beta_K)$ consisting of distinct $\beta_k\in \Fq$.  If $\bolda = (a_1,\ldots,a_K)\in \Fq^K$ is a message vector, define $u_{\bolda,\boldbeta}\in \Fq[z]^{<K}$ by the property $u_{\bolda,\boldbeta}(\beta_k) = a_k$ for all $k = 1,\ldots,K$.  The Lagrange Interpolation formula guarantees the existence of $u_{\bolda,\boldbeta}$, and the uniqueness is a consequence of the degree restriction. 

A generator matrix corresponding to this encoding map is given by
\begin{align}\label{equation:LagrangeMatrix}
&\boldG = \boldG(\boldalpha,\boldbeta)\triangleq \nonumber\\	&\begin{pmatrix}
\displaystyle\prod_{j\in[K]\setminus\{1\}}\frac{\beta_j-\alpha_1}{\beta_j-\beta_1} & \cdots & \displaystyle\prod_{j\in[K]\setminus\{1\}}\frac{\beta_j-\alpha_N}{\beta_j-\beta_1}\\
\displaystyle\prod_{j\in[K]\setminus\{2\}}\frac{\beta_j-\alpha_1}{\beta_j-\beta_2} & \cdots & \displaystyle\prod_{j\in[K]\setminus\{2\}}\frac{\beta_j-\alpha_N}{\beta_j-\beta_2}\\
\vdots&\vdots&\ddots&\vdots\\
\displaystyle\prod_{j\in[K]\setminus\{K\}}\frac{\beta_j-\alpha_1}{\beta_j-\beta_K} & \cdots & \displaystyle\prod_{j\in[K]\setminus\{K\}}\frac{\beta_j-\alpha_N}{\beta_j-\beta_K}
\end{pmatrix},
\end{align}
where we omit the notation~$\boldalpha,\boldbeta$ when clear from context.  If one defines $\beta_k = \alpha_k$ for $k = 1,\ldots,K$, then the above matrix $\boldG$ is in systematic form. The fact that~$\image \boldG(\boldalpha,\boldbeta)\triangleq\{ \bolda\boldG(\boldalpha,\boldbeta) \vert \bolda\in\bF_q^K \}$ equals~$\cR\cS_K(\boldalpha)$ for every~$\boldbeta$ might not be readily seen, but is easily proved in the following lemma.
    \begin{lemma}\label{lemma:LagrangeisRS}
    For every set~$\boldalpha\subseteq \bF_q$  of size~$N$, and every set~$\boldbeta\subseteq \bF_q$ of size~$K$, we have that
	$\image \boldG(\boldalpha,\boldbeta)=\cR\cS_K(\boldalpha)$.
\end{lemma}
\begin{proof}
	Let $\boldG = \boldG(\boldalpha,\boldbeta)$.  First, observe that~$\bolda \boldG=ev(u_{\bolda,\boldbeta})$ for every~$\bolda\in\bF_q^K$. This readily implies the inclusion~$\image \boldG\subseteq\cR\cS_K(\boldalpha)$, since $\cR\cS_K(\boldalpha)$ includes all evaluations at~$\boldalpha$ of all univariate polynomials of degree at most~$K-1$, one of which is~$u_{\bolda,\boldbeta}$. To prove the converse inclusion, let~$ev(f)\in\cR\cS_K(\boldalpha)$, and define~$\boldy=(f(\beta_1),\ldots,f(\beta_K))$.
	Since the polynomial~$u_{\boldy,\boldbeta}$ agrees with the polynomial~$f$ on~$\{\beta_i\}_{i=1}^K$, and since the degrees of both~$f$ and~$u_{\boldy.\boldbeta}$ are at most~$K-1$, if follows that~$f=u_{\boldy,\boldbeta}$, which implies that~$ev(u_{\boldy,\boldbeta})=ev(f)$, and hence~$ev(f)\in\image \boldG$.
\end{proof}

\subsection{Star Products}\label{section:starproduct}
Let $\boldx = (x_1,\ldots,x_N),\boldy= (y_1,\ldots,y_N)\in\Fq^N$.  Their \emph{star product} is defined to be
\[
\boldx\star\boldy = (x_1y_1,\ldots,x_Ny_N)\in \Fq^N
\]
Now let $\cC$ and $\cD$ be length $N$ linear codes over $\Fq$.  Their \emph{star product} $\cC\star\cD$ is another linear code of length $N$, defined to be
\[
\cC\star\cD = \spn_{\Fq}\{\boldc\star\boldd : \boldc \in \cC,\ \boldd \in \cD\}.
\]
Similarly, if $G\geq 1$ is any positive integer, we define
\[
\cC^{\star G} = \underbrace{\cC\star\cdots\star\cC}_{G\text{ times}} = \spn_{\Fq}\{\boldc_1\star\cdots\star\boldc_G : \boldc_g\in \cC\}.
\]

Reed-Solomon codes are especially well-behaved with respect to star products, since the star product of two evaluation vectors is the evaluation vector of the product of the two functions.  The following proposition essentially appears in \cite{StarProduct} and \cite{Karpuk}, but we include a proof here for the sake of completeness.
\begin{proposition}\label{proposition:starRS}
    Let $\cR\cS_K(\boldalpha)$ and $\cR\cS_T(\boldalpha)$ be Reed-Solomon codes of length $N$ with the same evaluation vector $\boldalpha$.  Then:
    \begin{itemize}
        \item[(i)] $\cR\cS_K(\boldalpha)\star\cR\cS_T(\boldalpha) = \cR\cS_{\min\{K+T-1,N\}}(\boldalpha)$, and \item[(ii)] $\cR\cS_K(\boldalpha)^{\star G} = \cR\cS_{\min\{G(K-1)+1,N\}}(\boldalpha)$.
    \end{itemize}
\end{proposition}
\begin{proof}
    Let $f,g\in \Fq[z]$, where $\deg(f)<K$ and $\deg(g)<T$.  Then $\deg(fg)<K+T-1$, and hence $\cR\cS_K(\boldalpha)\star\cR\cS_T(\boldalpha) \subseteq \cR\cS_{\min\{K+T-1,N\}}(\boldalpha)$. Conversely, every polynomial of degree less than $K+T-1$ is a linear combination of polynomials of the form $fg$, where $\deg(f)< K$ and $\deg(g)< T$.  This suffices to prove part (i).  Part (ii) is proved using part (i) and an easy induction argument on $G$.
\end{proof}

\subsection{The Shamir Secret Sharing Scheme}
The perfect privacy guarantees in the sequel can be seen as a special case of the Shamir secret sharing scheme~\cite{Shamir}.
By now a classic result, the Shamir secret sharing scheme allows~$N$ parties to share~$L$ secrets, such that sets of at most~$X$ parties cannot infer anything about the secrets, and sets of at least~$X+L$ parties can reconstruct all~$L$ secrets. The Shamir scheme relies on linear encoding of the following form. Let~$\bolds_1,\ldots,\bolds_L$ be the secrets, seen as column vectors over a finite field. An external trusted party generates~$X$ random column vectors~$\boldt_1,\ldots,\boldt_X$ of the same length as the secrets, performs linear encoding
\begin{align*}
    [\bolds_1,\ldots,\bolds_L,\boldt_1,\ldots,\boldt_X]\cdot \begin{pmatrix}\boldA\\\boldB\end{pmatrix}=(\boldy_1,\ldots,\boldy_N),
\end{align*}
where~$\boldA\in\bF_q^{L\times N}$ and~$\boldB\in \bF_q^{X\times N}$, and distributes the \textit{shares}~$\boldy_1,\ldots,\boldy_N$ to the parties. The following lemma is well-known, and will be most useful in the sequel.
\begin{lemma}\label{lemma:Shamir}
    If~$\boldB$ is an MDS matrix (i.e., if every~$X\times X$ submatrix of it is invertible), then
    \[
    I(\bolds_1,\ldots,\bolds_L;\boldy_{n_1},\ldots,\boldy_{n_X})=0
    \]
    for every subset~$\{n_1,\ldots,n_X\}\subseteq[N]$.
\end{lemma}

\subsection{Private Computation of Coded Data}\label{section:preliminariesPC}
We consider the problem of Private Computation on distributed storage systems of the following type; this follows a standard setup in the PIR literature, see \cite{RazanSalim,BanawanUlukusCoded,StarProduct}.  Let $\boldx_1,\ldots,\boldx_K\in \Fq^{M\times 1}$, and let $\cC$ be an $[N,K]$ code over $\Fq$ with generator matrix $\boldG\in \Fq^{K\times N}$.  Define vectors $\boldy_n\in \Fq^{M\times 1}$ for $n = 1,\ldots,N$ by
\begin{align}\label{equation:dataencoding}
    \left[\boldy_1\ \cdots\ \boldy_N\right] &= \boldX\cdot\boldG\mbox{, where}\nonumber\\
    \boldX&\triangleq \left[\boldx_1\ \cdots\ \boldx_K\right].
\end{align}
The vector $\boldy_n$ is stored on server $n$.  We refer to $\cC$ as the \emph{storage code}. 

Given the above setup, a user wishes to compute $\phi_b(\boldx_k)$ for some functions $\phi_1,\ldots,\phi_B$, for all $k = 1,\ldots,K$.  We assume that the functions $\phi_b$ all belong to some (necessarily finite-dimensional) vector space $\cS$ of functions $\Fq^{M\times 1}\rightarrow\Fq$.  To accomplish this goal, the user sends $S$ \emph{queries} $\rho_n^{(1)},\ldots,\rho_n^{(S)}\in\cS$ to the $n$'th server, who responds with the \emph{answers} $\rho_n^{(1)}(\boldy_n),\ldots,\rho_n^{(S)}(\boldy_n)$.  From all $NS$ answers, the user must be able to decode the desired function evaluations:
\[
H(\{\phi_b(\boldx_k)\}|\{\rho_n^{(s)}(\boldy_n)\}) = 0.
\]
It is useful to think of the above as happening over $S$ \emph{rounds} or \emph{iterations}, so that during the $s$'th round the user queries the servers with the functions $\rho_n^{(s)}$ and obtains the answers $\rho_n^{(s)}(\boldy_n)$.  Similarly, it is useful to think of the parameter $B$ as analogous to the block length of a file in traditional PIR.  We view the parameters $B$ and $S$ as free for the user to adjust to maximize their download rate.  Here the terms $\phi_b(\boldx_k)$ are random variables in the sense that the contents of the database are, to the user, unknown and therefore best treated as random. The terms $\rho_n^{(s)}(\boldy_n)$ are random variables in the sense that the queries $\rho_n^{(s)}$ are sampled according to some distribution employed by the user to preserve privacy.  We forego making this precise for the sake of readability. 

Our primary function space of interest is $\cS = \cP_G = \cP_{M,G}$, the space of polynomial functions of total degree at most $G$ from $\Fq^{M\times1}$ to $\Fq$.  Thus we in general have $\phi_b,\rho_n^{(s)}\in \cP_{G}$, for every~$b\in[B]$, $n\in[N]$, and~$s\in[S]$.

For any $T$-subset $\mathcal{T} = \{n_1,\ldots,n_T\}$ of $[N]$, we let $\rho_{\mathcal{T}}$ be the joint distribution of all $\rho_{n}^{(s)}$ for all $n\in \mathcal{T}$ and all $s = 1,\ldots,S.$  PC scheme has \emph{$T$-function-privacy} if
\begin{align*}
&I(\phi_1,\ldots,\phi_B;\rho_{\mathcal{T}}) = 0\quad\\
&\text{for all $T$-subsets $\{n_1,\ldots,n_T\}$ of $[N]$.}
\end{align*}
That is, a PC scheme has $T$-function-privacy if the identities of the functions $\phi_1,\ldots,\phi_B$ to be computed remain private even after any $T$ of the servers collude to attempt to deduce the identities of the $\phi_b$.

A PC scheme has \emph{$E$-data-privacy} if
\begin{align*}
&I(\boldx_1,\ldots,\boldx_K;\boldy_{n_1},\ldots,\boldy_{n_E}) = 0 \quad\\
 &\text{for all $E$-subsets $\{n_1,\ldots,n_E\}$ of $[N]$.}
\end{align*}
That is, the servers in the distributed storage system remain oblivious to the contents of the uncoded data, even if $E$ of them collude to attempt to deduce the identities of the $\boldx_k$.

A PC scheme is \emph{robust against $P$ stragglers} or \emph{unresponsive servers} if the user is still able to decode the values $\phi_b(\boldx_k)$ even if, during any round of the scheme, up to $P$ servers respond with an erasure symbol $?$ instead of the true answer $\rho_n^{(s)}(\boldy_n)$.  Similarly, a PC scheme is \emph{robust against $A$ adversaries} or \emph{Byzantine servers} if the user is still able to decode the values $\phi_b(\boldx_k)$ even if, during any round of the scheme, up to $A$ servers respond with an arbitrary element of $\Fq$ instead of the true answer $\rho_n^{(s)}(\boldy_n)$.  From a coding-theoretic perspective, having $P$ stragglers and $A$ adversaries simply means that during the $s$'th round, the user receives the total response vector
\[
(\rho_1^{(s)}(\boldy_1),\ldots,\rho_N^{(s)}(\boldy_N)) + \boldepsilon^{(s)}
\]
where $\boldepsilon^{(s)}$ is a vector containing at most $P$ erasure symbols $?$ and at most $A$ non-zero elements of $\Fq$.  Here the erasure symbol is understood to be absorbing with respect to addition: $x+? = ?$ for all $x\in \Fq$.  See \cite{SunJafarColluding,BanawanUlukusByzantine,Karpuk2} for more on PIR from systems with stragglers and adversaries.

Given a PC scheme, our principal metric of efficiency will be the \emph{download rate}, also referred to as the \emph{PC rate} or simply \emph{rate}, which is defined to be
\begin{equation}\label{ratedefn}
R = \frac{KB}{NS}.
\end{equation}
That is, the rate $R$ is the number of desired function evaluations $\phi_b(\boldx_k)$ the user obtains, divided by the total number of function evaluations $\rho_n^{(s)}(\boldy_n)$ downloaded.  From a strict Information-Theoretic point of view, it would be more correct to define the rate to be
\begin{equation}\label{ratedefn2}
R = \frac{H\left(\{\phi_b({\boldx_k)\}}\right)}{\sum_{n = 1}^N\sum_{s = 1}^SH\left(\rho_n^{(s)}(\boldy_n)\right)},
\end{equation}
which better accounts for potential dependencies between the variables $\phi_b(\boldx_k)$.  If all $\phi_b(\boldx_k)$ are independent for all $b,k$, and furthermore $H(\phi_b(\boldx_k)) = H(\rho_n^{(s)}(\boldy_n))$ for all $b,k,n,s$, then the two expressions in (\ref{ratedefn}) and (\ref{ratedefn2}) for the rate coincide.  This is indeed the case given reasonable independence conditions on the data vectors and the functions to be evaluated, but for the sake of compactness and readability we will ignore these subtleties and use (\ref{ratedefn}) as our definition of rate.

Much of the current literature on PIR concerns itself with establishing the \emph{capacity} $\mathsf{C}$ of a given PIR setup.  The capacity of a PIR problem is defined to be the supremum of all possible PIR rates.  As the current work is only concerned with explicit scheme constructions and not with establishing the capacity of any PC setup, we will mostly ignore this notion in the sequel.

We aggregate all of the important parameters of our system in Table \ref{tab:parameters}.

\begin{table}
	\begin{center}
		\small\begin{tabular}{|c|p{7cm}|}
			\hline
			$q$ & field size \\\hline
			$K$ & number of uncoded data vectors $\boldx_k$ \\\hline
			$M$ & length of vectors $\boldx_k$ \\\hline
			$G$ & degree of polynomial functions to be evaluated on the $\boldx_k$ \\\hline
			$N$  &  number of servers \\\hline
			$\mathcal{C}$  &  $[N,K]$ storage code \\\hline
			$T$ & number of colluding function-curious servers \\\hline
			$E$ & number of colluding data-curious servers \\\hline
			$A$ & number of adversaries / Byzantine servers \\\hline
			$P$ & number of stragglers / unresponsive servers \\\hline
			$S$ & number of rounds \\\hline
			$B$ & number of functions to be evaluated \\\hline
			$R$ & download rate of PC scheme \\\hline
		\end{tabular}
	\end{center}
	\caption{Important parameters used in Private Computation schemes.}	\label{tab:parameters}  
\end{table}

\subsection{Private Information Retrieval as Private Computation and a Remark on Upload Cost}\label{section:uploadcost}
The problem of Private Information Retrieval is a special case of the problem of Private Computation described above.  Indeed, in PIR one typically supposes that the data matrix $\boldX$ is composed of several row blocks $\boldx^i\in \Fq^{B\times K}$, which are the files of the system.  Setting the functions $\phi_b$ to be the $B$ coordinate projections corresponding to the $B$ coordinates of some $\boldx^i$, we see that computing $\phi_b(\boldx_k)$ for all $b=1,\ldots,B$ and all $k = 1,\ldots,K$ is equivalent to retrieving each row of $\boldx^i$, that is, downloading the file $\boldx^i$.  Since coordinate projections have degree $G = 1$, any PC scheme for privately computing polynomials of arbitrary degree $G$ specifies to a PIR scheme when we set $G = 1$ and choose the $\phi_b$ to all be distinct coordinate projections.

In Private Information Retrieval, one typically justifies ignoring the upload cost of a PIR scheme by assuming that the files $\boldx^i$ have entries in some field extension $\mathbb{K}$ of $\Fq$, while the storage code $\mathcal{C}$ and queries $\rho_n^{(s)}$ all remain defined over $\Fq$.   Provided that $[\mathbb{K}:\Fq]\gg0$, the upload cost is dominated by the download cost, which justifies using the download rate as the sole performance metric.  Generally, one can ignore the extension field $\mathbb{K}$ as all important operations and analysis occur over the base field $\Fq$.

Similar considerations allow one to ignore the upload cost in Private Computation.  One can suppose that the data matrix is defined over some large extension $\mathbb{K}$ of $\Fq$, and that the function space $\cP_G$ consists of polynomials with coefficients in $\Fq$ itself.  Communicating an arbitrary polynomial in $\cP_G$ costs $\sum_{i = 0}^G\binom{M}{i}$ elements of $\Fq$, and therefore the user's total upload cost will be $NS\sum_{i = 0}^G\binom{M}{i}$.  However, the total download cost is easily seen to be $NS\cdot [\mathbb{K}:\Fq]$ many elements of $\Fq$, and therefore if $\mathbb{K}$ satisfies $[\mathbb{K}:\Fq]\gg \sum_{i = 0}^G\binom{M}{i}$, one can ignore the upload cost.  We note that the upload cost increases very quickly in $G$, and therefore in practice this inequality may be satisfied only for small~$G$. 

Finally, we note that a simple relaxation of the privacy requirement can drastically reduce the upload cost. Notice that the expression $NS\sum_{i=0}^{G}\binom{M}{i}$ for upload cost stems from the expression~$NS\dim_{\bF_q}(\cP_G)$, where~$\dim_{\bF_q}(\cP_G)$ is the dimension of the set~$\cP_G$ as a subspace over~$\bF_q$. More generally, in order to communicate a function in a subspace~$\cP\subseteq \cP_G$, where~$\cP$ is known to all, one only has to communicate~$\dim_{\bF_q}(\cP)$ field elements.	Hence, one can determine any such subspace~$\cP$, either a priori or by communicating a sparse basis of it. Then, one can apply an identical scheme in which $\phi_b\in\cP$ for every~$b\in[B]$. The resulting privacy, however, will be restricted to~$\cP$, that is, the definition of~$T$-function privacy will be reduced to $$I(\phi_1,\ldots,\phi_B\vert \forall i\in[B],\phi_i\in\cP;\rho_\cT\vert \forall n_i\in\cT,\rho_{n_i}\in\cT)=0$$ for all~$T$-subsets~$\{n_1,\ldots,n_T\}$ of~$[N]$.

\section{A General Construction}\label{section:MainScheme}
In this section we present the main Private Computation scheme of the paper which applies for $N > (G(K+E-1)+T+P+2A)$. With the system parameters as in Table~\ref{tab:parameters}, the PC scheme has rate
\begin{align*}
R =\;& \frac{N - (G(K+E-1)+T+P+2A)}{N}\cdot\\ &\frac{K}{G(K+E-1)+1}.
\end{align*}
In what follows, we describe the data encoding procedure, then present the scheme in broad strokes.  In the third subsection, we present the scheme in detail, with an emphasis on the presentation of the first two rounds of the scheme for the sake of clarity.  We show that the scheme has the stated rate, has $E$-data privacy, has $T$-function privacy, and is robust against any $P$ stragglers and any $A$ adversaries.

\subsection{Data Encoding}
Let us first describe the data encoding procedure, which follows the Lagrange encoding of Reed-Solomon codes as in Section \ref{sec:lagrange_encode}.  Let $\boldx_k\in\Fq^M$ for $k = 1,\ldots,K$ be the $K$ data vectors, and let $\boldt_1,\ldots,\boldt_E\in\Fq^M$ be $E$ i.i.d.\ uniform random vectors.  Define $\boldX = \left[\boldx_1\cdots \boldx_K\ \boldt_1\cdots\boldt_E\right]$. Let $\boldbeta = (\beta_1,\ldots,\beta_{K+E})\in\Fq^{K+E}$ consist of $K+E$ distinct elements of $\Fq$.  Define the interpolation polynomial $u_{\boldX}(z) = u_{\boldX,\boldbeta}(z)$ by the property
\begin{align*}
u_{\boldX}(\beta_k) &= \boldx_k & \text{for }k &= 1,\ldots,K \\
u_{\boldX}(\beta_{K+e}) &= \boldt_e & \text{for }e &= 1,\ldots,E.
\end{align*}
By basic facts about polynomial interpolation, we have $\deg(u_{\boldX}(z))\leq K+E-1$.  Note that $u_\boldX(z)$ is more accurately described as a vector of polynomials of length $M$, that is, $u_\boldX(z)\in\Fq[z]^M$.  However, we continue to refer to it as a polynomial for simplicity.  We now choose an evaluation vector  $\boldalpha = (\alpha_1,\ldots,\alpha_N)\in \Fq^N$ where the $\alpha_n$ are all distinct and non-zero, and set
\[
\boldy_n = u_{\boldX}(\alpha_n)
\]
which is then stored on server $n = 1,\ldots,N$.  If we further have $\{ \beta_k \}_{k\in[K]}\cap \{ \alpha_n \}_{n\in[N]}=\varnothing$, then it is straightforward to show using Lemma \ref{lemma:Shamir} that we have $E$-data privacy. See \cite[Section IV.C]{Lagrange} for more details.  

\subsection{Basic Scheme Outline}  
To begin, define $N' \triangleq N - (P+2A)$, and let
\[
H \triangleq N'-(G(K+E-1)+T),\quad L \triangleq G(K+E-1)+1.
\]
Define $B$ and $S$ to be positive integers satisfying the equation
\[
BL = HS.
\]
The exact choice of $B$ and $S$ is not of crucial importance, so one can choose them to be minimal or simply set $S = L$ and $B = H$ for simplicity, which are the minimal solutions whenever $\gcd(L,H) = 1$.

The scheme will evaluate $\phi_b(\boldx_k)$ for $B$ functions $\phi_b\in\cP_G$.  This will be accomplished by downloading all coefficients of $\gamma_b(z)\triangleq\phi_b(u_{\boldX}(z))$, and then evaluating the $\gamma_b(z)$ on the $\beta_k$.  We have $\deg(\gamma_b(z))\leq G(K+E-1)$, and thus $L$ is the number of coefficients we need to download to completely determine a single $\gamma_b(z)$.  As the individual coefficients of $\gamma_{b}(z)$ will also play a role in the scheme construction, we define $\gamma_{b\ell}$ by
\[
\gamma_b(z) = \sum_{\ell = 0}^{L-1}\gamma_{b\ell}z^\ell
\]
The parameter $H$ is defined so that the user will download $H$ unique coefficients $\gamma_{b\ell}$ during each round of the scheme.

We define some $\zeta_b^{(s)}(z)\in\Fq[z,z^{-1}]$ for $b = 1,\ldots,B$ and $s = 1,\ldots,S$, whose precise nature will be made clear in the scheme construction.  We also choose\footnote{As mentioned in Subsection~\ref{section:uploadcost}, one can restrict the attention to some subspace~$\cP\subseteq \cP_G$ to reduce upload costs, in which case we must have~$\phi_b\in\cP$ for every~$b\in [B]$ and~$\psi_t^{(s)}\in\cP$ for every~$s$ and~$t$. The resulting query in~\eqref{equation:queryfunctions} will be in~$\cP$ as a linear combination. The guaranteed privacy will be restricted to~$\cP$, i.e., an attacker will not be able to learn anything about the functions~$\phi_1,\ldots,\phi_B$ other than them being in~$\cP$.} $\psi_t^{(s)}\in\cP_G$ to be i.i.d.\ uniform random elements, which are chosen anew during each round.  Now set
\[
\rho^{(s)} = \sum_{b = 1}^B\zeta_b^{(s)}(z)\phi_b + z^H\sum_{t = 1}^Tz^{t-1}\psi_t^{(s)}
\]
Given our evaluation points $\alpha_n\in\Fq$, we construct \emph{query functions} $\rho_n^{(s)}$, for $n = 1,\ldots,N$, the $n$'th of which is equal to $\rho^{(s)}$ evaluated at $z = \alpha_n$.  That is,
\begin{align}\label{equation:queryfunctions}
    \rho^{(s)}_n \triangleq \sum_{b = 1}^B\zeta_b^{(s)}(\alpha_n)\phi_b + \alpha_n^H\sum_{t = 1}^T\alpha_n^{t-1}\psi_t^{(s)}\in \cP_G.
\end{align}
Note that the coefficients~$\zeta_b^{(s)}(\alpha_n)$ and~$\alpha_n^{H+t-1}$ in~\eqref{equation:queryfunctions} are independent of the functions~$\phi_b$ and~$\psi_t^{(s)}$, and hence can be computed a priori. Here $\rho_n^{(s)}$ is transmitted to the $n$'th server during the $s$'th round, who responds with $\rho_n^{(s)}(\boldy_n)$.   We define the $s$'th \emph{response polynomial} $r^{(s)}(z)\in\Fq[z]$ to be
\[
\begin{aligned}
r^{(s)}(z) &\triangleq \rho^{(s)}(u_\boldX(z)) \\
&= \sum_{b = 1}^B\zeta_b^{(s)}(z)\phi_b(u_\boldX(z)) + z^H\sum_{t = 1}^Tz^{t-1}\psi_t^{(s)}(u_\boldX(z)) \\
&= \sum_{b = 1}^B\zeta_b^{(s)}(z)\gamma_b(z) + z^H\delta^{(s)}(z)
\end{aligned}
\]
where we define $\delta^{(s)}(z) \triangleq \sum_{t = 1}^T z^{t-1}\psi_t^{(s)}(u_\boldX(z))$.
Note that
\begin{equation}\label{noisepoly}
z^H\delta^{(s)}(z) \in \spn\{z^H,\ldots,z^{H+G(K+E-1)+T-1}\},
\end{equation}
and thus the evaluation vector of $z^H\delta^{(s)}(z)$, which contains only randomness, lives in a subspace of $\Fq^N$ of dimension $G(K+E-1)+T=N'-H$.

During the $s$'th round, the user observes the vector
\[
ev(r^{(s)}(z)) + \boldepsilon^{(s)}
\]
where $\boldepsilon^{(s)}$ is a vector which contains at most $P$ erasure symbols, coming from the $P$ stragglers, and at most $A$ arbitrary elements of $\Fq$, coming from the $A$ adversaries.    During each round, this will allow the user to decode the coefficients of $1,z,\ldots,z^{H-1}$ in $r^{(s)}(z)$, which will come from the terms in $\sum_{b = 1}^B\zeta_b^{(s)}(z)\gamma_b(z)$.  The general expression for $r^{(s)}(z)$ will be of the form
\begin{align*}
r^{(s)}(z) =\;& \underbrace{a_{-C}^{(s)}z^{-C} + \cdots + a_{-1}^{(s)}z^{-1}}_{\text{$a_i^{(s)}$ known from previous rounds}}\\
& + \underbrace{a_0^{(s)} + \cdots + a_{H-1}^{(s)}z^{H-1}}_{\text{$a_i^{(s)}$ decoded during $s$'th round}} + \underbrace{z^H\delta^{(s)}(z)}_{\text{noise}}
\end{align*}
for some $C$ which can depend on the round index $s$.  The assumption that $\alpha_n\neq 0$ for all $n$ guarantees that $ev(r^{(s)}(z))$ is a well-defined element of $\Fq^N$. Strictly speaking, $r^{(s)}(z)$ is not a polynomial as it will contain negative powers of~$z$, but the distinction is not relevant to the scheme construction, and so we refer to it as a polynomial for simplicity. 

During the $s$'th round, the decoding process roughly proceeds as follows.  The user subtracts off the evaluation vectors of $a^{(s)}_{-C}z^{-C},\ldots,a^{(s)}_{-1}z^{-1}$ from what they receive, and are left with an element of a Reed-Solomon code with parameters $[N,N',P+2A+1]$.  This allows them to correct the $P$ erasures and $A$ errors in the vector $\boldepsilon^{(s)}$.  After this decoding, they may obtain the elements $a_0^{(s)},\ldots,a^{(s)}_{H-1}$ as the coefficients of the evaluation vectors of $1,z,\ldots,z^{H-1}$, which are linearly independent from the evaluation vector of $z^H\delta^{(s)}$.

The scheme construction will guarantee that the sets $A^{(s)} = \{a_i^{(s)}\ |\ i = 0,\ldots,H-1\}$ of coefficients decoded during round $s$ each consist of $H$ unique coefficients of the polynomials $\gamma_b(z)$.  That is, $A^{(s)}\cap A^{(t)}=\varnothing$ for $s\neq t$, and since each has size $H$, the user will have decoded $HS = BL$ unique coefficients of the $B$ polynomials $\gamma_b(z)$ at the end of the scheme.  Since this is all of the coefficients of these polynomials, the user can reconstruct all $\gamma_b(z)$ entirely, and therefore compute
\[
\phi_b(\boldx_k) = \phi_b(u_\boldX(\beta_k)) = \gamma_b(\beta_k)
\]
for all $b = 1,\ldots,B$ and all $k = 1,\ldots,K$.  The PC rate is therefore
\begin{align*}
R = \frac{KB}{NS} = \;&\frac{N-(G(K+E-1)+T+P+2A)}{N}\cdot\\
& \frac{K}{G(K+E-1)+1}
\end{align*}
as claimed. 

\subsection{Construction of $\zeta_b^{(s)}(z)$ and Decoding} It remains to define the functions $\zeta_b^{(s)}(z)$ and describe the decoding process.  To begin, define the matrix
\[
\zeta(z) = \left(\zeta_b^{(s)}(z)\right)_{\substack{ 1\leq s\leq S \\ 1\leq b \leq B }} \in \Fq[z,z^{-1}]^{S\times B}
\]
and let $\zeta^{(s)}(z)$ be the $s$'th row of $\zeta(z)$.  We start by defining integers $Q_1$ and $U_1$ by using Euclidean division to write
\[
H = Q_1L + U_1,\quad Q_1 = \left\lfloor H/L \right\rfloor\quad\text{and}\quad 0\leq U_1< L.
\]
For the sake of presentation, we first define $\zeta^{(1)}(z)$ and $\zeta^{(2)}(z)$ and demonstrate the decoding process during the first two rounds of the scheme.  We then give a recursive definition of $\zeta_b^{(s)}(z)$ which works for every $s = 2,\ldots,S$.

\subsubsection{Round $s = 1$}  Define $\zeta^{(1)}(z)$ to be
\[
\zeta^{(1)}(z) = \begin{pmatrix}
1 & z^L & \cdots & z^{(Q_1-1)L} & z^{Q_1L} & 0 & \cdots & 0
\end{pmatrix}
\]
or equivalently,
\[
\zeta_b^{(1)}(z) = \left\{\begin{array}{cl}
z^{(b-1)L} & 1\leq b\leq Q_1+1 \\
0 & Q_1+2\leq b \leq B
\end{array}\right.
\]
The response polynomial $r^{(1)}(z)$ is then of the form
\begin{align*}
r^{(1)}(z) =\;& \gamma_1(z) + z^L\gamma_2(z) + \cdots + z^{(Q_1-1)L}\gamma_{Q_1}(z) +\\& z^{Q_1L}\gamma_{Q_1+1}(z) + z^H\delta^{(1)}(z)
\end{align*}
First, note that
\begin{align*}
\deg(z^{Q_1L}\gamma_{Q_1+1}(z))&\leq Q_1L + G(K+E-1)\\
&\leq H+G(K+E-1) \leq N'-1,
\end{align*}
that alongside~\eqref{noisepoly} implies that the evaluation vector of $r^{(1)}(z)$ lives in $\cR\cS_{N'}(\boldalpha)$, which is an MDS code with parameters $[N,N',P+2A+1]$ and can thus correct $P$ erasures and $A$ errors contained in the vector $\boldepsilon^{(s)}$.  After the erasure/error correction, the user is left with the evaluation vector of $r^{(1)}(z)$, which determines $r^{(1)}(z)$ completely.  Since $\deg(\gamma_b(z))\leq L$ for all $b$, the coefficients of $1,z,\ldots,z^{H-1}$ of the above allow the user to decode all of the coefficients of $\gamma_1(z),\ldots,\gamma_{Q_1}(z)$, and the first $U_1$ coefficients of $\gamma_{Q_1+1}(z)$.  That is, $A^{(1)} = \{\gamma_{10},\ldots,\gamma_{Q_1+1,U_1-1}\}$.  

\subsubsection{Round $s = 2$} Let $Q_2$ be the largest integer such that $-U_1+(Q_2-Q_1)L \leq  H$, and define $\zeta^{(2)}(z)$ by
\[
\zeta^{(2)}_b(z) = \left\{\begin{array}{cl}
0 & 1 \leq b\leq Q_1 \\
z^{-U_1+(b-(Q_1+1))L} & Q_1+1\leq b\leq Q_2 + 1 \\
0 & Q_2+2\leq b \leq B.
\end{array}\right.
\]
The response polynomial $r^{(2)}(z)$ is then of the form

\begin{alignat*}{2}
r^{(2)}(z) &=& z^{-U_1}\gamma_{Q_1+1}(z) + z^{-U_1+L}\gamma_{Q_1+2}(z) + \cdots \\
&~&+ z^{-U_1+(Q_2-Q_1)L}\gamma_{Q_2+1}(z) + z^H\delta^{(2)}(z) \\
&=& \underbrace{\gamma_{Q_1+1,0}z^{-U_1} + \cdots + \gamma_{Q_1+1,U_1-1}z^{-1}}_{\text{coefficients known from round $s = 1$}} +\\
&~&\underbrace{\begin{matrix}
 	\gamma_{Q_1+1,U_1} + \cdots + \gamma_{Q_2+1,0}z^{-U_1+(Q_2-Q_1)L} +\\
 	 \cdots + 	\gamma_{Q_2+1,U_2-1}z^{-U_1+(Q_2-Q_1)L+U_2-1}
 	\end{matrix}}_{\text{$H$ coefficients decoded during round $s=2$}} \\
&~&+ \underbrace{z^H\eta^{(2)}(z)}_{\text{noise}},
\end{alignat*}

where $U_2$ is chosen so that $-U_1+(Q_2-Q_1)L + U_2 - 1 = H - 1$, that is, so that the first $U_2$ coefficients of $\gamma_{Q_2+1}(z)$ are downloaded during round $2$.  Here $\eta^{(2)}(z)$ is a polynomial which incorporates the noise polynomial $\delta^{(2)}(z)$ and the monomials of $\gamma_{Q_2+1}(z)$ whose coefficients are \emph{not} downloaded during the $s$'th round.  One sees easily that 
\begin{align*}
\deg\left(z^{-U_1 + (Q_2-Q_1)L}\gamma_{Q_2+1}(z)\right)&\leq H + G(K+E-1)\\
&\leq N'-1
\end{align*}
from which we see that $\deg(z^H\eta^{(2)}(z))\leq N' - 1$ as well.

The coefficients $\gamma_{Q_1+1,0},\ldots,\gamma_{Q_1+1,U_1-1}$ are exactly the first $U_1$ coefficients of $\gamma_{Q_1}(z)$, and are known to the user from round $s = 1$.  The user can therefore subtract the vectors $\gamma_{Q_1+1,0}ev(z^{-U_1}),\ldots,\gamma_{Q_1+1,U_1-1}ev(z^{-1})$ from the response vector in round $2$.  After this subtraction, the user is left with the vector

\begin{align*}
\underbrace{ev\begin{pmatrix}
			\gamma_{Q_1+1,U_1} +  \cdots + \\
		 \gamma_{Q_2+1,U_2-1}z^{-U_1+(Q_2-Q_1)L+U_2-1} +\\ z^H\eta^{(2)}(z)
	\end{pmatrix}}_{\in \cR\cS_{N'}(\boldalpha)} + \epsilon^{(2)}
\end{align*}
where $\epsilon^{(2)}$ contains $P$ erasure symbols and $A$ non-zero elements of $\Fq$.  Because $\cR\cS_{N'}(\boldalpha)$ is an MDS code with parameters $[N,N',P+2A+1]$, the user can correct these erasures and errors.  After this erasure/error correction, the user further decodes the coefficients $A^{(2)} = \{\gamma_{Q_1+1,U_1},\ldots,\gamma_{Q_2+1,U_2-1}\}$ as the coefficients of the evaluation vectors of $1,z,\ldots,z^{H-1}$.

\subsubsection{Round $s = 2,\ldots,S$} We generalize the construction of the $\zeta_b^{(2)}(z)$ to the following recursive construction of $\zeta_b^{(s)}$ for all $s = 2,\ldots,S$.  Having already defined $U_{s-1}$ and $Q_{s-1}$, we define $Q_s$ to be the maximal integer such that $-U_{s-1}+(Q_s-Q_{s-1})L\leq H$, and then set
\[
\zeta_b^{(s)}(z) = \left\{\begin{array}{cl}
0 & 1\leq b\leq Q_{s-1} \\
z^{-U_{s-1} + (b-(Q_{s-1}+1))L} & Q_{s-1}+1\leq b\leq Q_s + 1 \\
0 & Q_s+2\leq b\leq B.
\end{array}\right.
\]
Lastly, define $U_s$ by the equation $-U_{s-1} + (Q_s-Q_{s-1})L + U_s - 1 = H - 1$.  In round $s - 1$, the user decoded the first $U_{s-1}$ coefficients of $\gamma_{Q_{s-1}}(z)$.  The response polynomial $r^{(s)}(z)$ can therefore be written as
\begin{alignat*}{2}
r^{(s)}(z) &=&& \underbrace{\gamma_{Q_{s-1},0}z^{-U_{s-1}} + \cdots + \gamma_{Q_{s-1},U_{s-1}-1}z^{-1}}_{\text{coefficients known from round $s - 1$}} \\
&~&& \underbrace{\begin{matrix}
	+\gamma_{Q_{s-1},U_{s-1}} + \cdots + \gamma_{Q_s,0}z^{-U_{s-1}+(Q_s-Q_{s-1})L} +\\ \cdots +  \gamma_{Q_s,U_s-1}z^{-U_{s-1} + (Q_s-Q_{s-1})L + U_s - 1}+
	\end{matrix}}_{\text{$H$ coefficients decoded during round $s$}} \\
 &~&&\underbrace{z^H\eta^{(s)}(z)}_{\text{noise}}.
\end{alignat*}
As in rounds $s = 1,2$, the user first subtracts the vectors corresponding to negative powers of $z$, the coefficients of which have been decoded in round $s - 1$.  As in rounds $s = 1,2$, one shows easily that $\deg(z^H\eta^{(s)}(z))\leq N'-1$, hence what is left after this subtraction is an element of $\cR\cS_{N'}(\boldalpha)$ plus a vector $\boldepsilon^{(s)}$ with $P$ erasure symbols and $A$ non-zero elements of $\Fq$.  After performing erasure/error correction in this Reed-Solomon code, the user may now decode the set $A^{(s)} = \{\gamma_{Q_{s-1},U_{s-1}},\ldots,\gamma_{Q_s,U_s-1}\}$ as the coefficients of the evaluation vectors of $1,z,\ldots,z^{H-1}$.  We see that during the $s$'th round, the user obtains all coefficients of $Q_s-Q_{s-1}$ distinct polynomials $\gamma_b(z)$, and $U_s$ coefficients of a single additional polynomial.

After $S$ rounds, the user has decoded all of the sets $A^{(s)}$ for $s = 1,\ldots,S$.  Since they are all clearly disjoint and of size $H$, the user has obtained $HS = BL$ coefficients of the $\gamma_b(z)$, which suffices to reconstruct all of these polynomials.  The scheme construction is therefore complete.

\subsection{Proof of $T$-function Privacy}
Since a function in~$\cP_G$ can be seen as a vector over~$\bF_q$, it follows from~\eqref{equation:queryfunctions} that the query functions are of the following form.
\begin{align*}
    \begin{bmatrix}\rho_1^{(s)}\\\vdots\\\rho_N^{(s)}\end{bmatrix}^\top=\begin{bmatrix}\phi_1\\\vdots\\\phi_B\\\psi_1^{(s)}\\\vdots\\\psi_T^{(s)}\end{bmatrix}^\top
    \begin{pmatrix}
        \zeta_1^{(s)}(\alpha_1) & \ldots & \zeta_1^{(s)}(\alpha_N)\\
        \vdots & \ddots & \vdots \\
        \zeta_B^{(s)}(\alpha_1) & \ldots & \zeta_B^{(s)}(\alpha_N)\\
        \alpha_1^{H} & \ldots & \alpha_N^H \\
        \alpha_1^{H+1}  & \ldots & \alpha_N^{H+1} \\
        \vdots & \ddots & \vdots \\
        \alpha_1^{H+T-1}  & \ldots & \alpha_N^{H+T-1} \\
    \end{pmatrix}.
\end{align*}
It is readily verified by Lemma~\ref{lemma:Shamir} that $I(\phi_1,\ldots,\phi_B;\rho_{n_1}^{(s)},\ldots,\rho_{n_T}^{(s)}) = 0$ for any $T$-subset $\mathcal{T}$ of $[N]$, and any single $s = 1,\ldots,S$.  Since the $\psi_t^{(s)}$ are chosen independently between rounds, the result follows.

\section{Example}
In this section we illustrate some of the subtleties of the scheme construction through a detailed example for the parameters $N = 14$, $K = 2$, $G = 2$, $T = 1$, $P = 1$, $A = 1$, $E = 2$, which achieves
a PC rate of
\begin{align*}
R =&\; \frac{N-(G(K+E-1)+T+P+2A)}{N}\cdot\\ &\;\frac{K}{G(K+E-1)+1} = \frac{4}{49}.
\end{align*}
The auxiliary parameters are given by $N' = 11$, $H = 4$, $L = 7$, $B = 4$, and $S = 7$.  The data encoding operates by choosing two i.i.d.\ uniform random vectors $\boldt_1$, $\boldt_2$ and distinct elements $\beta_1,\ldots,\beta_4\in\Fq$, and defining the interpolating polynomial $u_\boldX(z)$ by the conditions
\[
u_\boldX(\beta_1) = \boldx_1,\ u_\boldX(\beta_2) = \boldx_2,\ u_\boldX(\beta_3) = \boldt_1,\ u_\boldX(\beta_4) = \boldt_2.
\]
We have $\deg(u_\boldX(z))\leq 3$.  Choosing some $\alpha_1,\ldots,\alpha_{14}\in\Fq$ to all be non-zero and distinct from $\beta_1$ and~$\beta_2$, we encode by setting $\boldy_n = u_\boldX(\alpha_n)$ for all $n = 1,\ldots,14$.

The scheme will compute, over $S = 7$ rounds, the function values $\phi_b(\boldx_k)$ for $B = 4$ quadratic functions $\phi_b$.  The polynomials $\gamma_b(z) = \phi_b(u_\boldX(z))$ whose coefficients will be downloaded have degree $\deg(\gamma_b(z))\leq 6$, and hence are each completely determined by $L = 7$ coefficients.

The query functions $\rho^{(s)}_n$ take the form
\[
\rho^{(s)}_n = \sum_{b = 1}^B\zeta_b^{(s)}(\alpha_n)\phi_b + \alpha_n^4\psi_1^{(s)}
\]
where $\psi_1^{(s)}$ is a uniform random element of $\cP_2$.  The functions $\zeta_b^{(s)}(z)$ are given by
\[
\zeta(z) = \left(\zeta_b^{(s)}(z)\right)_{\substack{1\leq s\leq 7 \\1\leq b\leq 4 }} = \begin{pmatrix}
1 & 0 & 0 & 0 \\
z^{-4} & z^3 & 0 & 0 \\
0 & z^{-1} & 0 & 0 \\
0 & z^{-5} & z^2 & 0 \\
0 & 0 & z^{-2} & 0 \\
0 & 0 & z^{-6} & z \\
0 & 0 & 0 & z^{-3}
\end{pmatrix}
\]
from which we can compute the response polynomials $r^{(s)}(z)$ to be
\begin{align*}
r^{(1)}(z) =&\;  \gamma_{10} + \gamma_{11}z + \gamma_{12}z^2 + \gamma_{13}z^3 + z^4\eta^{(1)}(z) \\
r^{(2)}(z) =&\;   \gamma_{10}z^{-4} +   \gamma_{11}z^{-3} +  \gamma_{12}z^{-2} +  \gamma_{13}z^{-1} \ +\  \gamma_{14} +\\
&\; \gamma_{15}z + \gamma_{16}z^2 + \gamma_{20}z^3 + z^4\eta^{(2)}(z) \\
r^{(3)}(z) =&\;  \gamma_{20}z^{-1}  \ +\  \gamma_{21} + \gamma_{22}z + \gamma_{23}z^2 +\\
&\; \gamma_{24}z^3 + z^4\eta^{(2)}(z) \\
r^{(4)}(z) =&\;  \gamma_{20}z^{-5} +  \gamma_{21}z^{-4} +  \gamma_{22}z^{-3} +  \gamma_{23}z^{-2} + \\
&\; \gamma_{24}z^{-1} \ +\  \gamma_{25} + \gamma_{26}z + \gamma_{30}z^2 + \\
&\;\gamma_{31}z^3 + z^4\eta^{(4)}(z) \\
r^{(5)}(z) =&\;  \gamma_{30}z^{-2} +  \gamma_{31}z^{-1}  \ +\  \gamma_{32} + \gamma_{33}z + \gamma_{34}z^2 +\\
&\; \gamma_{35}z^3 + z^4\eta^{(5)}(z) \\
r^{(6)}(z) =&\;  \gamma_{30}z^{-6} +  \gamma_{31}z^{-5} +  \gamma_{32}z^{-4} +  \gamma_{33}z^{-3} +\\
&\;  \gamma_{34}z^{-2} +  \gamma_{35}z^{-1}  \ +\  \gamma_{36} + \gamma_{40}z + \gamma_{41}z^2 +\\
&\; \gamma_{42}z^3 + z^4\eta^{(6)}(z) \\
r^{(7)}(z) =&\;  \gamma_{40}z^{-3} + \gamma_{41}z^{-2} +  \gamma_{42}z^{-1}  \ +\   \gamma_{43} + \gamma_{44}z +\\
&\; \gamma_{45}z^2 + \gamma_{46}z^3 + z^4\eta^{(7)}(z)
\end{align*}
In the first round, the polynomial $r^{(1)}(z)$ has degree $\leq N'-1 = 10$, hence its evaluation vector lives in a Reed-Solomon code with parameters $[14,11,4]$, allowing the user to correct the erasure and error from the straggler and adversary.  The user then decodes $\gamma_{10},\ldots,\gamma_{13}$.  In all successive rounds, it is clear from the above expressions that all coefficients of negative powers of $z$ are decoded in previous rounds, and can thus be subtracted from the response vector.  After this subtraction, what is left is in a Reed-Solomon code with parameters $[14,11,4]$, which can correct $P = 1$ erasure and $A = 1$ error.  The user performs this erasure/error correction, and then decodes the coefficients of $1,z,z^2,z^3$ in every round.

Having decoded all $\gamma_{b\ell}$ for $b = 1,\ldots,4$ and $\ell = 0,\ldots,6$, the user can compute $\phi_b(\boldx_k) = \gamma_b(\beta_k)$ for $b = 1,\ldots,4$ and $k = 1,2$.  The user has computed $KB = 8$ function evaluations while downloading $NS = 98$ elements of $\Fq$ total, hence the rate of the scheme is
\[
R = \frac{KB}{NS} = \frac{8}{98} = \frac{4}{49}
\]
as claimed.

\section{Improved PC rate for systematic encoding}\label{section:SystematicImprovement}

In this section it is shown that whenever~$P=A=E=0$ and~$N\le 2G(K-1)+T+1$, a larger PC rate is achievable by systematic encoding. The choice of the specific systematic code is not important, so one can simply assume that the data is encoded with systematic Lagrange encoding (i.e., where~$\beta_i=\alpha_i$ for all~$i\in[K]$ in~\eqref{equation:LagrangeMatrix}). The scheme in this section appeared in~\cite{Karpuk}.

First, let the storage code be $\cC = \cR\cS_K(\boldalpha)$.  Define~$\cD=\cR\cS_T(\boldalpha)$, called the \textit{retrieval code} of the scheme, and let~$\cE\triangleq \cC^{\star G}\star \cD$. It follows from Lemma~\ref{lemma:LagrangeisRS} and Proposition~\ref{proposition:starRS} that~$\cE=\cR\cS_{G(K-1)+T}(\boldalpha)$. Denote the minimum distance of~$\cE$ by~$D_\cE$, and let~$F\triangleq D_\cE-1=N-G(K-1)-T$. Next, let~$B$ and~$S$ be the minimal integers that satisfy~$\min\{F,K\}S=KB$ (i.e., $S=\frac{\lcm(K,\min\{F,K\})}{\min\{F,K\}}$ and $B=\frac{\lcm(K,\min\{ F,K \})}{K}$), and choose subsets of~$[K]$ --
\begin{align}\label{equation:queriesI}
    \cI^{(1)} &= \cI^{(1,1)}\cup \cI^{(1,2)}\cup \ldots \cup \cI^{(1,B)}\nonumber\\
    \cI^{(2)} &= \cI^{(2,1)}\cup \cI^{(2,2)}\cup \ldots \cup \cI^{(2,B)}\nonumber\\
    &\vdots\nonumber\\
    \cI^{(S)} &= \cI^{(S,1)}\cup \cI^{(S,2)}\cup \ldots \cup \cI^{(S,B)},
\end{align}
such that in each row the sets in the union are pairwise disjoint, such that~$|\cI^{(s)}|=\min\{F,K\}$ for every~$s\in[S]$, and such that for every~$b\in[B]$ we have that~$\cup_{s=1}^S \cI^{(s,b)}=[K]$.  The process of choosing the sets~$\{ \cI^{(s,b)} \}_{(s,b)\in[S]\times[B]}$ in~\eqref{equation:queriesI} is very simple, and similar constructions have appeared in \cite{StarProduct,Karpuk}.  For our current purposes, the choice of these sets is best illustrated by the following example. 
\begin{example}
    Assume that~$F=4$ and~$K=6$, which implies that~$S=3$ and~$B=2$. Consider the following matrix
    \begin{align*}
        \begin{pmatrix}
            1 & 1 & 1 & 1 &   & \\ 
            2 & 2 &   &   & 1 & 1 \\
              &   & 2 & 2 & 2 & 2 \\
        \end{pmatrix},
    \end{align*}
    which naturally corresponds to the sets
    \begin{align*}
        \cI^{(1,1)}& =\{1,2,3,4\} & \cI^{(1,2)}&=\varnothing\\
        \cI^{(2,1)}& =\{5,6\} & \cI^{(2,2)}&=\{1,2\}\\
        \cI^{(3,1)}& =\varnothing & \cI^{(3,2)}&=\{3,4,5,6\} .
    \end{align*}
One verifies immediately that the sets $\{\mathcal{I}^{(s,b)}\}_{(s,b)\in[S]\times[B]}$ have the stated properties.
\end{example}

To continue with the scheme construction, let~$Q$ be the dimension of~$\cP_G$, and fix a basis~$\psi^1,\ldots,\psi^Q$ of~$\cP_G$. In round~$s\in[S]$, the user employs fresh randomness in order to choose codewords~$\boldd^1,\ldots,\boldd^Q\in\cD$, and defines
\begin{align}\label{equation:randomNoise}
        \psi_n = \boldd^1(n)\cdot \psi^1+\cdots+\boldd^Q(n)\cdot \psi^Q
\end{align}
for every~$n\in[N]$, where~$\boldd^i(n)$ is the~$n$'th entry of~$\boldd^i$. Then, the user defines the queries
\begin{align*}
	\rho_n^{(s)} = 
	\begin{cases}
		\psi_n+\phi_b & \substack{\mbox{if there exists }b\in[B]\\\mbox{such that }n\in \cI^{(s,b)}}\\
		\psi_n & \mbox{else}\\
	\end{cases}.
\end{align*}

Intuitively, the set~$\cI^{(s)}\subseteq [N]$ contains the indices of the~$\boldy_i$'s for which some~$\phi_b(\boldy_i),b\in[B]$ is retrieved during the~$s$'th round. In addition, for every~$b\in[B]$, the set~$\cI^{(s,b)}$ contains the indices of the~$\boldy_i$'s such that~$\phi_b(\boldy_i)$ is retrieved during the~$s$'th round; formally, we have that
\begin{align*}
    \cI^{(s,b)}\triangleq \{ t\in[N]~\vert~ \phi_b(\boldy_t)\mbox{ is retrieved in round }s \}.
\end{align*}

The~$T$-function-privacy relies on the following lemma, analogues of which were fully proved in~\cite[Theorem~1]{Karpuk} and in~\cite[Theorem~8]{StarProduct}, among others.
\begin{lemma}\label{lemma:PerfectPrivacy}
    If~$\cD$ is an MDS code then the scheme is~$T$-function-private.
\end{lemma}

In turn, the response vector is of the form
\begin{align}\label{equation:responseSystematic}
	\boldrho^{(s)}(\boldY)=\underbrace{[\psi_1(\boldy_1),\ldots,\psi_N(\boldy_N)]}_{\triangleq \boldpsi(\boldY)}+\;\boldv^{(s)},
\end{align}
where
\begin{align*}
	(\boldv^{(s)})_n=
	\begin{cases}
		\phi_b(\boldx_n) & \substack{\mbox{if there exists }b\in[B]\\\mbox{such that }n\in \cI^{(s,b)}}\\
		0 & \mbox{else}
	\end{cases},
\end{align*}
and notice that~$\phi(\boldy_n)=\phi_b(\boldx_n)$ merely since~$\cI^{(s)}\subseteq[K]$, and since the encoding is systematic.

\begin{lemma}\label{lemma:Correctness}
    The codeword~$\boldpsi(\boldY)\triangleq [\psi_1(\boldy_1),\ldots,\psi_N(\boldy_N)]$ is in~$\cE$.
\end{lemma}

\begin{proof}
    According to~\eqref{equation:randomNoise}, one can readily verify that
    \begin{alignat*}{1}
        &[\psi_1(\boldy_1),\ldots,\psi_N(\boldy_n)]=\\&\left[\sum_{j=1}^Q\boldd^j(1)\cdot\psi^j(\boldy_1),\ldots,\sum_{j=1}^Q\boldd^j(N)\cdot \psi^{j}(\boldy_N)\right]\\
        &=\sum_{j=1}^Q\left[ \boldd^j(1)\cdot \psi^j(\boldy_1),\ldots,\boldd^j(N)\cdot \psi^j(\boldy_N)
        \right]\\
        &= \sum_{j=1}^Q  [\psi^j(\boldy_1),\ldots,\psi^j(\boldy_N)]\star\boldd^j,
    \end{alignat*}
    and hence it remains to prove that~$[\psi^j(\boldy_1),\ldots,\psi^j(\boldy_N)]\in\cC^{\star G}$ for every~$j\in[Q]$. To this end, assume that the chosen basis~$\psi^1,\ldots,\psi^Q$ consists only of monomials, and let~$\psi^j=X_1^{a_1}X_2^{a_2}\cdots X_M^{a_M}$ for some~$j\in[Q]$. Denoting~$\boldy_n \triangleq [y_{n,1},y_{n,2},\ldots,y_{n,M}]$ for every~$n\in[N]$, which implies that~$ [y_{1,m},y_{2,m},\ldots,y_{N,m}]\in\cC $ for every~$m\in[M]$, implies that
    \begin{alignat*}{2}
        [\psi^j(\boldy_1),\ldots,&\psi^j(\boldy_N)]=\\
        &\left[ \prod_{m=1}^M y_{1,m}^{a_m}, \prod_{m=1}^M y_{2,m}^{a_m},\ldots,\prod_{m=1}^M y_{N,m}^{a_m} \right]\\
        &=[y_{1,1},y_{2,1},\ldots,y_{N,1}]^{\star a_1}\star\ldots\star\\
        & \qquad[y_{1,M},y_{2,M},\ldots,y_{N,M}]^{\star a_M}\\
        &\in \cC^{\star a_1}\star \ldots \star \cC^{\star a_M}\subseteq \cC^{\star G}.\qedhere
    \end{alignat*}
\end{proof}

Since~$\boldpsi(\boldY)$ is a codeword in~$\cE$, and since the Hamming weight of~$\boldv^{(s)}$ is at most~$|\cI^{(s)}|\le F=D_\cE-1$, it follows 
that the vector~$\boldv^{(s)}$ can be extracted from~$\boldrho^{(s)}(\boldY)$ by an error correction algorithm for~$\cE$. Therefore, according to the definition of the sets~$\{ \cI^{(s,b)} \}_{(s,b)\in[S]\times [B]}$, it follows that the user can retrieve~$\{ \phi_b(\boldx_k) \}_{(k,b)\in [K]\times [B]}$, which concludes the correctness of the scheme, whose PC rate is
\begin{align}\label{equation:RateSystematic}
    \frac{KB}{NS}=\frac{\min\{ F,K \}}{N}.
\end{align}
Recall that for the parameters~$P=A=E=0$, the PC rate of the general scheme~\eqref{mainrate} is
\begin{align}\label{equation:mainrateSystematic}
    \frac{N-G(K-1)-T}{N}\cdot\frac{K}{G(K-1)+1}.
\end{align}
Therefore, in order to compare~\eqref{equation:RateSystematic} with~\eqref{equation:mainrateSystematic}, notice that if~$N> 2G(K-1)+T+1$, then since~$G\ge 1$, it follows that
\begin{align}\label{equation:RateCompare}
    F\ge G(K-1)+1\ge K,
\end{align}
and thus~\eqref{equation:RateSystematic} reduces to~$K/N$. Moreover, \eqref{equation:RateCompare} also implies that
\begin{align*}
    \frac{F}{N}\cdot \frac{K}{G(K-1)+1}\ge \frac{K}{N},
\end{align*}
and hence for this parameter regime we have that~\eqref{equation:mainrateSystematic} is superior to~\eqref{equation:RateSystematic}. Conversely, if~$N\le 2G(K-1)+T+1$, then~$F\le G(K-1)+1$, and therefore 
\begin{align*}
    \frac{F}{N}\cdot \frac{K}{G(K-1)+1}\le \frac{K}{N}.
\end{align*}
Also, since~$K\le G(K-1)+1$, we have that
\begin{align*}
    \frac{F}{N}\cdot \frac{K}{G(K-1)+1}\le \frac{F}{N},
\end{align*}
and hence
\begin{align*}
    \frac{F}{N}\cdot \frac{K}{G(K-1)+1}\le \frac{\min\{F,K\}}{N},
\end{align*}
which implies the superiority of the systematic scheme for~$N\le 2G(K-1)+T+1$.

\begin{figure}
    \centering
    \includegraphics[scale=0.5]{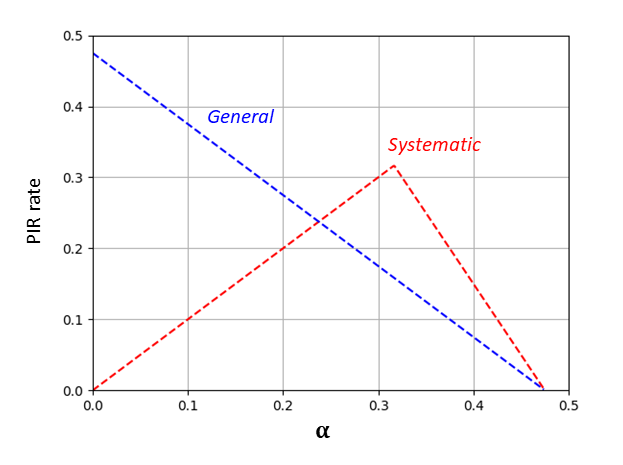}
    \caption{Asymptotic comparison of PC rates for various values of~$\alpha = K/N$ between the general scheme (Subsection~\ref{section:MainScheme}) and the systematic one (Subsection~\ref{section:SystematicImprovement}) for~$G=2$ and~$\beta= T/N =0.05$.}
    \label{fig:sysVSnonsys}
\end{figure}

To compare~\eqref{equation:RateSystematic} and~\eqref{equation:mainrateSystematic} asymptotically, assume that~$N,K,$ and~$T$ approach infinity, and yet~$\alpha\triangleq K/N$ and~$\beta\triangleq T/N$ remain constants. The resulting PC rates are~$\frac{1-\beta}{G}-\alpha$ for the general scheme, and~$\min\{ 1-\alpha G-\beta,\alpha \}$ for the systematic one, and the necessary condition~$N\ge G(K-1)+T+1$ translates to~$\alpha G+\beta\le 1$. An example for~$G=2$ and~$\beta=0.05$ is given in Figure~\ref{fig:sysVSnonsys}.

\section{Discussion and open problems}
In this paper we leveraged a recent notion from coded computing and obtained a robust private computing scheme for computing polynomials of arbitrary degree on encoded data, that can also guarantee privacy of the data.  The scheme also accounts for servers which are stragglers or adversaries, and is robust against servers colluding to attempt to deduce the identity of the functions to be evaluated. Many open questions remain, the most prominent of which is the capacity question, i.e., constructing schemes for this scenario with provably optimal rate.  However, as mentioned in Section~\ref{section:intro}, the capacity question of PIR is largely open, even without extending its scope to private computation. Thus, proving optimality is likely to be very difficult, even in restricted cases where some of the system parameters are zero. Another possible research direction is coming up with a proper notion of perfect privacy for \textit{real-valued} computations, which applies directly without embedding into a finite field.

\section*{Acknowledgements}
The first author would like to thank Prof.~Jehoshua Bruck for many helpful discussions. The second author would like to thank Razane Tajeddine and Oliver Gnilke for constructive and helpful conversations regarding the results of the current work.

\end{document}